\documentclass[final]{siamltex}

\usepackage{soul,color,algorithm,amsfonts}
\usepackage{amsfonts,amssymb,epsf,algorithmic}

\usepackage{url}
\usepackage{multirow}
\usepackage{subfigure,epsfig}

\newtheorem{lem}{Lemma}
\newtheorem{cor}{Corollary}
\newtheorem{thm}{Theorem}
\newtheorem{prop}{Proposition}

\usepackage{enumitem}      

\usepackage{color}

\def\half{\mbox{$\frac{1}{2}$}}

\def\boldalpha{\mbox{\boldmath $\alpha$}}

\def\0{\mathbf{0}}
\def\1{\mathbf{1}}

\def\eps{\epsilon}

\begin{document}

\title{Algorithms and Complexity of Range Clustering}

\author{Dorit S. Hochbaum
 }
\thanks{Department of Industrial Engineering and Operations Research,
University of California, Berkeley, CA 94720
(hochbaum@ieor.berkeley.edu). This author's research was supported
in part by NSF award No. CMMI-1200592.
}




\maketitle

\begin{abstract}
We introduce a novel criterion in clustering that seeks clusters with limited \emph{range} of values associated with each cluster's elements. In clustering or classification the objective is to partition a set of objects into subsets, called clusters or classes, consisting of similar objects so
that different clusters are as dissimilar as possible. We propose a number of objective functions that employ the range of the clusters as part of the objective function.  Several of the proposed objectives mimic objectives based on sums of similarities. These objective functions are motivated by image segmentation problems, where the
diameter, or range of values associated with objects in each cluster, should be small. It is
demonstrated that range-based problems are in general easier, in terms of their
complexity, than the analogous similarity-sum problems.  Several of the
problems we present could therefore be viable alternatives to existing
clustering problems which are NP-hard, offering the advantage of efficient algorithms.
\end{abstract}

\pagestyle{myheadings} \thispagestyle{plain} \markboth{D. S. Hochbaum
}{Range Clustering}
\section{Introduction}
The typical clustering and classification problem is to partition a set of objects into subsets, called clusters, so that each subset consists of ``similar" elements, and the different clusters are as dissimilar as possible. We introduce novel clustering criteria that seek clusters with limited \emph{range} of values associated with each cluster's objects (or elements).  Each of the objects to be clustered has a scalar value associated with it and the {\em range} of a cluster is the difference between the maximum and the minimum values of the elements in the cluster.   We show that for a partition to $k$ clusters the problem of minimizing the maximum range of a cluster and the problem of the ranges of the clusters are solvable in polynomial time. Other problems explored here are the $k$-normalized range sum, the $k$-range cut and the $k$ normalized range cut, and we provide particularly efficient algorithms for the first two and demonstrate that the last one is NP-complete.

A good model for separation between clusters is the minimum cut, applied to the
graph with nodes representing objects, and edges between pairs of nodes
associated with weights of the similarity between the respective nodes.  The
minimum $2$-cut problem is to find a partition of the graph into two
nonempty components so that the sum of similarity weights on edges with one
endpoint in one cluster and the second endpoint in the other cluster, is
minimum.  Thus the similarity between the two resulting clusters of the bipartition is minimum.
This notion of cut extends to multiple clusters, as in the minimum $k$-cut
problem where the objective is to find a partition into $k$ nonempty components
so that the sum of weights of edges with endpoints in different components, or the inter-similarity, is
minimum \cite{GH}. The concept of a cut plays an important role in image
segmentation problems, where the goal is to partition the image into meaningful
objects.  The problems of $k$-range cut and $k$ normalized range cut combine the range objective with a form of the cut objective.  The problems' formulations are given in Tables \ref{tb:2-formulate} and \ref{tb:k-formulate}, and the complexity results and run times of the algorithms are given in Table \ref{tb:res}.

The motivation for our study of range problems originates in image segmentation. In a typical image segmentation set-up there are similarity weights assigned to each pair of {\em adjacent} pixels (which are the objects for the image
clustering/segmentation), \cite{Hoc01,SM}.  However, it is often the case that
each pixel has in addition some scalar value associated with it, such as its color
intensity, or its texture (computed with respect to a neighborhood of the pixel).  The goal
in this type of clustering is then not only to have the pixels similar to each other, but also to
have the scalar values associated with pixels in the same cluster close enough to each other.  This is the case, for instance, in segmenting knee cartilage as in \cite{HFBG}, where pixels of cartilage tissue have a distinctive texture.  The
goal then is to ensure that all pixels 
within each segmented object have only a limited variability in their range of
texture values.


We devise a family of {\em range-based} clustering problems that are analogous to commonly
considered goals in clustering. We demonstrate that, in general, range problems are easier to solve (in terms of complexity) than their respective total similarity problems. One such problem is the NP-hard \emph{normalized cut} problem \cite{SM}, defined below in Equation \ref{eq:normcut}.   In contrast, we demonstrate here that an analogous range objective problem similar to normalized cut is polynomial-time solvable.

The term range-clustering was previously used in \cite{phan15}.   However, the context there is
to provide an improved computation of similarities, rather than to generate meaningful clusters as is the case here.  We believe that here is the first time that the concept of range is utilized as a clustering criterion.

To formalize the discussion and problem definitions, we introduce relevant graph notation and other preliminaries.

\subsection{Notations and preliminaries}\label{ref:prelims}
Let $G=(V,E)$ be an undirected graph where the nodes of the graph correspond to elements (also referred to as objects) to be clustered.  There are edge weights $w_{ij}$ associated with each edge $[i,j]\in E$ representing the ``similarity" of nodes $i$ and $j$.  A similarity weight is in turn equivalent to
the penalty of not assigning the respective pair of elements $v_i$ and $v_j$ to the same cluster. Higher similarity is associated with higher weights. A bi-partition of a graph is called a {\em cut}, $(S,\bar{S})=\{ [i,j] |i\in S, j\in \bar{S} \}$, where $\bar{S}=V\setminus S $.   The {\em capacity of a cut} $(S,\bar{S})$ is the sum of the weights of the edges, with one endpoint in $S$ and the other in $\bar{S}$:
$
 C(S,\bar{S})= \sum _{i\in S, j\in \bar{S}, [i,j]\in E}w_{ij}.
$
A minimum capacity cut, or minimum cut, is a bipartition of the nodes into two non-empty sets $(S,\bar{S})$ that minimizes $C(S,\bar{S})$.  A bipartition resulting from a minimum cut has $S$ and $\bar{S}$ as dissimilar as possible in terms of the sum of similarities between their elements.


Two versions of the minimum cut problem are $2$-cut and $s,t$-cut. The former is the bipartition of $V$
into two non-empty sets, $S$ and its complement $\bar{S}$, $(S,\bar{S})$.  For designated nodes $s,t\in V$, $(S,\bar{S})$ is said to be an $s,t$-cut, if
$s\in S$ and $t\in \bar{S}$.  A minimum $2$-cut or $s,t$-cut are those bipartitions that minimize $C(S,\bar{S})$.  When there is no ambiguity, we will simply refer to the {\em min-cut}.

A partition of the set of elements, $V$, into more than two
sets, $k$-partition, is a collection of $k \geq 3$ disjoint non-empty sets
${\{ S_i \}}_{i=1}^k$ so that $\cup_{i=1}^k S_i = V$.
A $k$-partition is denoted by $(S_1,\ldots,S_k)$.  We refer to a partition into $k$ clusters as {\em $k$-clustering}.

The {\em weighted degree} of node $i$ is the sum of weights adjacent to $i$,
$d_i= \sum _{j|[i,j]\in E}w_{ij}$.  The weight of a subset of nodes $B \subseteq V$, referred to as the \emph{volume} of $B$, is the sum of weighted degrees of nodes in $B$, $d(B)=\sum_{i:v_i \in B} d_i$. 

The concept of {\em shrinking} of nodes is utilized here.  Consider a graph $G=(V,E)$  with edge weights $w_{ij}$ associated with each edge $[i,j]\in E$, and a specific pair of nodes $p$ and $q$.  The process of shrinking node $p$ into $q$ is to remove node $p$ from the graph, and appending to node $q$ all edges formerly adjacent to $p$.  This can be done in two equivalent methods.  In the first, we let $N(p)= \{ v | [p,v]\in E\}$, remove the set of edges $\{ [p,v] |v\in N(p)\}$ and add the set of edges $\{ [q,v] |v\in N(p)\}$ with their respective weights. The second method is to simply add an edge $[p,q]$ of infinite capacity (weight).  This ensures that $p$ and $q$ are always together in the same set when the minimum cut criterion is used.

Input graphs for range-clustering include node weights.  These are distinct scalar values associated with the nodes of $V$,  $\{\alpha_1,\ldots,\alpha_n\}$, where $\alpha_i$ are rational numbers.  These scalars are necessarily rational and of finite number of significant digits (so the input is finite). The assumption of distinct scalar values is important for the algorithms presented here.  This is because equal values of the scalars may lead to an exponential number of equally valued solutions, all of which may have to be explored by the algorithms.  This assumption of distinct scalar values is however shown next to hold without loss of generality.

If there are equal valued scalars, a standard {\em perturbation process} is applied:  This is done by adding different powers of a small enough $\eps$ to each of the equal values.  (Such process is applied in linear programming algorithms to avoid degeneracy and cycling.)  For instance, the value of $\eps$ can be selected to be the smallest resolution of any $\alpha _i$.  That is, multiply all $\alpha _i$ by a large enough number, say $M$, so they are all integers, and then set $\eps =\half$.  For $\alpha _i=\alpha _j$, we set their perturbed values to $\tilde{\alpha _i }= \alpha _i +\eps ^i$ and $\tilde{\alpha _j }= \alpha _j +\eps ^j$.  The perturbed values are then all distinct, and it is easy to see that an optimal solution to the perturbed problem, in terms of the range, is one of the optimal solutions to the unperturbed problem.  If the problem involves similarity weights as well, the values of $\eps$ would depend on the smallest resolution among the scalar values of $\alpha$ as well as the values of $w_{ij}$.

Therefore, it is assumed without loss of generality that $\alpha_1<\alpha _2<\ldots <\alpha_n$.  We let $\alpha _i$ be associated with the element $v_i\in V$ so that $v_1$ is the element with the smallest value $\alpha_1$ and $v_n$ is the element with the largest value $\alpha_n$.


For any non-empty set $B\subseteq V$  the maximum, minimum and range of $B$ are defined
as:
\begin{eqnarray*}
max(B)&=&\max_{i:v_i\in B}\{\alpha_i\}\\
min(B)&=&\min_{i:v_i\in B}\{\alpha_i\}\\
range(B)&=&max(B)-min(B).
\end{eqnarray*}
Note that the range of an empty set is undefined, and is not relevant here.  This is because we are interested only in partitions of the set of elements into non-empty clusters of bounded range.  Hence, an empty cluster is not considered. Note also that the range of a singleton is $0$.

We say that a set $S$ is a {\em subset of an interval} $I=[a,b]$ if $S\subseteq \{j| a\leq \alpha _j \leq b\}$ and $min(S)=a$, $max(S)=b$.  We denote ``$S$ a subset of $I$" by $S\subseteq I$.

The complexity model used here is the real computation model which allows arithmetic
operations on real numbers, regardless of the number of significant digits, to count as a single operation, \cite{blumcomplexity}.

We next introduce our range-based clustering problems and review related known clustering objectives that utilize similarities.   Highlights of the differences in complexity between the range based and known clustering objective are discussed.

\subsection{Range-based clustering problems and related clustering problems utilizing cuts and similarities} \label{relwork}

We introduce here a new collection of range-based problems and discuss related known clustering problems that utilize cuts and similarities.
First we present $2$-clustering range problems.  The list of the names and formulations of the range-based clustering problems for bipartitioning problems is given in Table \ref{tb:2-formulate}.

The simplest case of bipartition range-based problems considered is the {\sf min range sum}.
This objective function seeks to minimize the range of $S$ and $\bar{S}$ simultaneously, and it is shown here to be solvable in polynomial time.  A weighted version of the problem is the {\sf min weighted range sum} which permits one to emphasize the limited range of $S$, more than that of $\bar{S}$.  The {\sf min weighted range sum} is also solved in polynomial time, as shown in Section \ref{Prob1}. The {\sf min-max range problem} is to minimize the bottleneck range between the two sets of the bipartition.

Many commonly used clustering models utilize the notion of minimum cut.  The input to such problems is a graph $G=(V,E)$ and similarity weights associated with the edges.  Bipartition clustering is to partition the set of elements $V$ into two
non-empty disjoint sets, $S$ and its complement $\bar{S}$.   The capacity, or weight, of the cut between $S$ and $\bar{S}$, $C(S,\bar{S})$, signifies the degree of similarity between $S$ and $\bar{S}$. To generate a set $S$ that is highly dissimilar to its complement one seeks a minimum cut partition to $S$ and $\bar{S}$ that minimizes $C(S,\bar{S})$.  This in turn also maximizes the total similarity within the two sets.

It has long been observed that a minimum cut in a graph with edge similarity weights tends to create a bipartition that has one side very small in size, containing a singleton in extreme cases \cite{WL}.  This is so because the number of edges between a single node and the rest of the graph tends to be much smaller than between two comparable-sized sets. To correct for such unbalanced bipartitions, Shi and Malik, in the context of image segmentation \cite{SM}, proposed the {\em normalized cut} as an alternative criterion to minimum cut. The normalized cut (NC) optimization problem is to find a bipartition of $V$, $(S,\bar{S})$, minimizing:
\begin{eqnarray}
{\mbox {  (NC)}}\ \ \min_{\varnothing
\subset S \subset V}  \ \frac{C(S,\bar{S})}{d(S)}+\frac{C(S,\bar{S})}{d(\bar{S})}.
\end{eqnarray} \label{eq:normcut}
The normalized cut problem (NC) was shown to be NP-hard in \cite{SM} by a reduction from set partitioning.  Because set partitioning is weakly NP-hard, this only proves that normalized cut is at least weakly NP-hard.  The problem is however strongly NP-hard with a reduction from the {\em balanced cut problem},  which is sketched below for an easier problem. The essence of the difficulty of NC derives from the fact that in the objective function of NC, the one ratio term with the smaller value of $d()$ is at least $\frac{1}{2}$ of the objective value. Therefore, this objective function drives the segment $S$ and its complement to be approximately of equal sizes.  Indeed, it is shown in \cite{Hoc10} that the problem of minimizing the first term of NC, $\min_{\varnothing
\subset S \subset V}  \ \frac{C(S,\bar{S})}{d(S)}$, is polynomial time solvable.

The following problem is also a form of normalizing the cut with respect to the size of the sets.  This problem is associated with finding the graph expander ratio and is known to be NP-hard.  It is referred to as {\em size-normalized cut}. 
\begin{eqnarray}
{\mbox {  (size-NC)}}\ \ \min_{\varnothing
\subset S \subset V}  \ \frac{C(S,\bar{S})}{|S|}+\frac{C(S,\bar{S})}{|\bar{S}|}
\end{eqnarray} \label{eq:expander}
Note that like NC, the objective function of size-NC drives the segment $S$ and its complement to be approximately of equal sizes.

The balanced cut problem is to find a bipartition cut of minimum weight such that both sets in the bipartition contain half of the nodes of the graph, (or $\min \{|S|, |\bar{S}|\} \leq c |V|$ for a constant $c\in (0,\half)$).  This problem is also known as {\em minimum cut into bounded sets}, proved NP-complete in \cite{GJS76}.  For an edge weighted graph $G$ of total sum of edge weights equal to $M$, we scale all the edge weights by $M$.  Since the numerator then is very small, an optimal solution is attained for half the nodes in source set and the other half in sink set, while minimizing the cut value, which is an optimal solution for the respective balanced cut problem.  Therefore the size-NC is a strongly NP-hard problem.

\begin{table} [htb]
\begin{center} 
\begin{tabular}{|l|l|} \hline
{\bf Problem Name} & {\bf Formulation}
\\ \hline \hline
{\sf min range sum} &  $\min_{\emptyset
\subset S \subset V}  \ range(S)+range(\bar{S})$  \\
{\sf min max range} & $\min_{\emptyset
\subset S \subset V}  \ \max \{ range(S),range(\bar{S}) \}$  \\
{\sf min normalized range sum} &  $\min_{\varnothing
\subset S \subset V}  \ \frac{range(S)}{f(|S|)}+\frac{range(\bar{S})}{f(|\bar{S}|)}$\\
{\sf min range cut} & $ \min_{\varnothing
\subset S \subset V}  \ range(S)+range(\bar{S})+C(S,\bar{S})$ \\
{\sf min normalized range cut}&  $\min_{\varnothing
\subset S \subset V}  \ \frac{range(S)}{f(|S|)}+\frac{range(\bar{S})}{f(|\bar{S}|)}+C(S,\bar{S})$\\
\hline \hline
\end{tabular}
\vspace{0.1in}\caption{Formulations of bipartition range clustering problems.} \label{tb:2-formulate}
\vspace{-0.11in}
\end{center}
\end{table}

The {\sf min normalized range sum} is a range-based objective analogous to NC and size-NC.
Unlike the normalized cut problem, this problem is shown here to be polynomial time solvable for $f(|S|)$ monotone non-decreasing in $|S|$.

For the next two range problems, {\sf min range cut} and {\sf min normalized range cut}, the input consists of a graph $G=(V,E)$, edge weights (similarities) $w_{ij}$ for all $[i,j]\in E$, and scalars $\alpha _i$ associated with each node $i\in V$.
The problems {\sf min range cut} and {\sf min normalized range cut} are generated by adding a minimum cut capacity term to the objective functions of {\sf min range cut} and {\sf min normalized range cut} respectively.
The {\sf min range cut} problem is shown here to be polynomial time solvable, whereas the {\sf min normalized range cut} is proved to be NP-hard, even for $f(|S|)=|S|$. Thus in the latter problem, the addition of the minimum cut term to the objective changes its status from a polynomial problem to an NP-hard problem.

A sample of $k$-clustering problems for $k\geq 3$ that utilize similarities include:
\begin{enumerate}
\item    The minimum $k$-cut problem seeks a $k$ partition so that the total
    similarity between all $k$ clusters is minimized.   This similarity is measure by the sum of weights of edges that have endpoints in different clusters.  The minimum $k$-cut
    problem is NP-hard, but can be solved in polynomial time for fixed $k$ \cite{GH}.  The minimum $2$-cut problem is a polynomial time solvable special case of the $k$-cut where $k=2$.


\item Clustering into $k$ subsets of constrained size:  Feo et al., in \cite{FGK}, considered this problem in the context of  VLSI design application and gave a heuristic solution.  This problem is NP-hard even
    for partition into {\em two} clusters. In this formulation, the requirement that the size of the sets is bounded, say by a fraction (such as half) of the total number
    of elements, makes the problem equivalent to the balanced cut problem, which is
    NP-hard.
\item
    The extension of normalized cut (NC) to $k$-clustering is $\min _{(S_1,S_2 \ldots S_k)} \sum _{i=1}^k \frac{ C(S_i,\bar{S_i})} { d(S_i) }$, which is obviously NP-hard as well.
\end{enumerate}
The respective formulations of $k$-partition range-based problems for $k\geq 3$, that generalize the $2$-clustering objectives, are presented in Table \ref{tb:k-formulate}.  The {\sf min $k$-normalized range sum} problem is analogous to the $k$-normalized cut.   The problems of {\sf min $k$-range cut} and {\sf min $k$-normalized range cut} are generated by adding the $k$-cut objective to the respective objectives of {\sf min $k$ range sum} and {\sf min $k$-normalized range sum} respectively.

\begin{table} [htb]
\begin{center} 
\begin{tabular}{|l|l|} \hline
{\bf Problem Name} & {\bf Formulation}
\\ \hline \hline
{\sf min $k$ range sum} &  $\min_{(S_1,\ldots,S_k)}  \ \sum_{i=1}^k range(S_i)$  \\
{\sf min max $k$ range} & $\min_{(S_1,\ldots,S_k)} \left( \max_{i\in \{1\ldots k\}} range(S_i)\right)$  \\
{\sf min $k$-normalized range sum} &  $\min_{(S_1,\ldots,S_k)}  \ \sum_{i=1}^k \frac{range(S_i)}{f(|S_i|)}$ \\
{\sf min $k$-range cut} & $\min_{(S_1,\ldots,S_k)}  \ \sum_{i=1}^k range(S_i) + \sum_{i=1}^{k-1} \sum_{j=i+1}^k C(S_i,S_j)$  \\
{\sf min $k$-normalized range cut}&  $\min_{(S_1,\ldots,S_k)}  \ \sum_{i=1}^k \frac{range(S_i)}{|S_i|} + \sum_{i=1}^{k-1} \sum_{j=i+1}^k C(S_i,S_j)$ \\
\hline \hline
\end{tabular}
\vspace{0.1in}\caption{Formulations of $k$-partition range clustering problems for $k\geq 3$.} \label{tb:k-formulate}
\vspace{-0.11in}
\end{center}
\end{table}

Among our results for $k$-clustering range problems, the {\sf min max k range}, {\sf min k range sum} and {\sf min k-normalized range sum} problems are shown to be polynomial-time solvable.  The {\sf min k-normalized range cut}
problem is proved to be NP-hard.  The {\sf min $k$ range cut} problem is polynomial for fixed $k$, but proved to be NP-hard for general $k$. The complexity results are summarized in Table \ref{tb:res} under the assumption that the input scalar values are given sorted. This assumption is made in order to highlight the fact that the running time for solving some of the problems is faster than that required to sort the values, $O(n\log n)$.

\begin{table} [htb]
\begin{center} 
\begin{tabular}{|l|l|l|} \hline
{\bf min problem} & {\bf Bipartition}, $k=2$ & {\bf $k$-Partition}, $k\geq 3$
\\ \hline \hline
{\sf $k$-range sum} &  Polynomial time, $O(n)$ &  Polynomial time, $O(n )$\\
{\sf max $k$-range }&  $O(\log n)$ &   Polynomial time, $O(n\log^3 n )$  \\
{\sf $k$-normalized range sum}$^*$  &  Polynomial time, $O(n)$&  Polynomial time, $O(n^2 k)$\\
{\sf $k$-range cut} &  Polynomial time, $O(mn^2 \log {\frac{n^2}{m}}))$ &  Polynomial $O(n^{k^2} )$ for $k$ fixed\\
 &   &  NP-complete for general $k$\\
 {\sf $k$-normalized range cut}&  NP-complete & NP-complete\\

\hline \hline
\end{tabular}
\vspace{0.1in}\caption{Summary of the range clustering problems and algorithmic results. $^*$$f(|S|)$ is monotone non-decreasing in $|S|$.} \label{tb:res}
\vspace{-0.11in}
\end{center}
\end{table}

\noindent
{\bf Paper overview:} The remainder of the paper is organized as follows:  Sections \ref{Prob1}, \ref{minmax} and \ref{normRange} present polynomial time algorithms solving the {\sf min range sum}, the {\sf min max range}, and the {\sf min normalized range sum} problems, respectively.  In Section \ref{Prob2} we describe a polynomial time algorithm solving the {\sf min range cut} problem which uses a parametric cut procedure.   In Section \ref{Prob3} the {\sf min normalized range cut} problem is proved to be NP-hard.  Section \ref{k-seg} provides algorithms and associated complexity results for the respective $k$-clustering problems for $k\geq 3$.
Concluding remarks are given in Section \ref{conc}.


\section{The minimum range sum problem}\label{Prob1}

We let $Z(S)$ be the objective value for the {\sf min range sum} problem with
the set $S$, and let $S^*$ be an optimal solution to the problem:
\begin{eqnarray*}
Z(S^*) = \min_{\varnothing \subset S \subset V} Z(S) =  \min_{\varnothing
\subset S \subset V} range(S) + range(\bar{S}).
\end{eqnarray*}
Using the definition of $range(\cdot)$, this problem may be written as:
\begin{eqnarray*}
\min_{\varnothing \subset S \subset V} max(S)-min(S)+max(\bar{S})-min(\bar{S})
 \label{objective1} \end{eqnarray*}
Recall that the nodes are indexed according to their respective values of $\alpha$, $\alpha _1<\ldots < \alpha _n$.  Therefore node $v_1$ is associated with $\alpha _1$ and node $v_n$ with $\alpha _n$.  Without loss of generality, we can assume $\alpha_1 \in S$; thus $\min(S) = \alpha_1$. The next lemma proves that in an optimal solution the element with the largest value,
$v_n$, belongs to $\bar{S}$:

\begin{lem}\label{lem:n_in_S_bar}
In every optimal solution to the {\sf min range sum} problem, $v_1 \in S$ and
$v_ n \in \bar{S}$.
\end{lem}

\begin{proof} If $v_1\in \bar{S}$, then $S$ and $\bar{S}$ may be interchanged with no change
to the value of the objective function.  When $v_1, v_n \in S$, the objective
value can be no lower than ${\alpha_n - \alpha_1}$.  However the solution $S'=V\backslash v_n$ and $\bar{S}'=v_n$ has an
objective value that is at most ${\alpha_{n-1} - \alpha_1}$ with
$v_n\in\bar{S}'$. So a solution that does not contain both $v_1$ and $v_n$ in $S$ is strictly better. 
\end{proof}

Therefore, we can assume that $max(\bar{S})=\alpha_n$ and
$min(S)=\alpha_1$. A more general lemma proves that there exists an optimal solution with $max(S) \leq
min(\bar{S})$:
\begin{lem}\label{lem:maxS}
There is an optimal solution to the {\sf min range sum} problem with $max(S) \leq
min(\bar{S})$.
\end{lem}
\begin{proof} Suppose not.  Let $(S,\bar{S})$ be a feasible solution
with $max(S)=\alpha _j > \alpha_{\ell}= min(\bar{S})$.  Thus $1<\ell <j<n$.  Construct another feasible solution by letting
$S'=S\backslash\{v_{j},v_{j-1},...,v_{\ell}\}$ and $\bar{S'}=
\bar{S}\cup\{v_{j},v_{j-1},...,v_{\ell}\}$.  The solution $(S',\bar{S'})$ is
feasible since $S'$ and $\bar{S'}$ are  non-empty.  Note that the maximum and
minimum elements of $\bar{S'}$ are the same as those of $\bar{S}$, and thus the
range for both is the same.
However, the objective value $Z(S')$ can only be smaller than $Z(S)$:
   \[
\hspace{.4in}\begin{array}{ll} Z(S') & =max(S')-min(S')+max(\bar{S'})-min(\bar{S'})=
\alpha _{\ell -1}-\alpha
_1-(\alpha_n-\alpha _{\ell}) =\\
& max(S)- (\alpha_j-\alpha _{\ell
-1})-min(S)+max(\bar{S})-min(\bar{S})=\\
& Z(S)- (\alpha_j-\alpha _{\ell -1})\leq
Z(S).
\end{array}\]

Thus there is  an optimal solution in which $range(S)$ and $range(\bar{S})$  form
two non-overlapping intervals. \end{proof}

\begin{figure}[ht]
  \begin{center}
  \scalebox{0.5}{
      \epsfig{figure =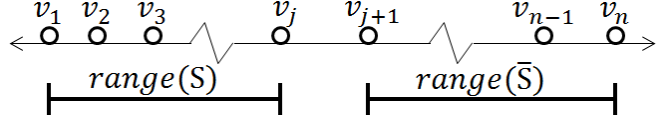}
      }
      \end{center}
  \caption{The elements of $V$ arranged on the real line with $range(S)$ and
  $range(\bar{S})$ for $S=\{v_1,v_2,...,v_j\}$ and $\bar{S}=\{v_{j+1},v_{j+2},...,v_n\}$. }\label{numberline}
\end{figure}
Consider the elements of $V$ arranged on the real line, where each element
$v_i$ is placed at position $\alpha_i$.   We observe that the largest gap between two consecutive elements plays a role in the {\sf min range sum} partition:
\begin{cor}\label{cor:range_sum_optimal}
If $p^*$ corresponds to the index of the largest value in the set of gaps $\mathcal{G}=\{\alpha_{p+1}-\alpha_{p}|, p=1,\ldots,n\}$, then an optimal solution to the {\sf min range sum}
problem is $S=\{v_{1},v_{2},...,v_{p^*}\}$ and
$\bar{S}=\{v_{{p^*+1}},v_{{p^*+2}},...,v_{n}\}$ with an optimal value of
$\alpha_{p^*}-\alpha_{1}+\alpha_{n}-\alpha_{{p^*+1}}$.
\end{cor}
\begin{proof}
It was shown in Lemma \ref{lem:n_in_S_bar} that  $min(S)=\alpha_1$ and
$max(\bar{S})=\alpha_n$.  Thus the objective is,
$Z(S^*)=\min_{\varnothing
\subset S \subset V} \  max(S) - \alpha_1 + \alpha_n - min(\bar{S})
$.
This problem is then equivalent to the maximization problem:
$\max_{\varnothing \subset S \subset V}  \ min(\bar{S})-max(S)$
which is to maximize the gap between the largest element of $S$ and  the
smallest element of $\bar{S}$ as stated. 
\end{proof}

Finding the largest gap requires $O(n)$ arithmetic operations and comparisons. This is generalized for $k$-partitions in Section \ref{prob1k} with the same run time, albeit with a more involved algorithm.

We comment that the algorithm for {\sf min range sum} can be extended to a weighted version of the problem, {\em weighted minimum range sum}. The {\em weighted minimum range sum} problem, for $\gamma \in (0,1)$, is:
\begin{eqnarray*}
\min_{\varnothing \subset S \subset V} Z(S) =  \min_{\varnothing
\subset S \subset V} range(S) + \gamma range(\bar{S}).
\end{eqnarray*}
This weighted problem is also solved in $O(n)$ time.  This is because Lemma \ref{lem:maxS} holds also for the weighted variant.  The algorithm is to compare the $n$ values of $\alpha _n -\alpha _{p+1} + \gamma ( \alpha _p - \alpha _1)$ and select the smallest.

\section{Min max range}\label{minmax}
Recall that the problem of min max range is $\min_{\emptyset
\subset S \subset V}  \ \max \{ range(S),range(\bar{S}) \}$.  This problem is solved relatively easily:  Let $\alpha (\half )=\frac{\alpha _1+ \alpha _n}{2}$.  Next, identify the index $i^*$ such that,
$$\alpha _{i^*} \leq \alpha (\half ) < \alpha _{i^*+1}.$$
We now let $S=\{1,\ldots ,i^*\}$ and $\bar{S}=\{i^*+1,\ldots ,n\}$. The minimum of $\max \{ range(S), range(\bar{S})\}$ is attained for $\max \{\alpha _{i^*}-\alpha _1, \alpha _n -\alpha _{i^*}\}$.
Since finding $i^*$ in the sorted array can be done with binary search in $O(\log n)$ steps, this is the complexity of solving the min max range problem.

\section{Min normalized range sum}\label{normRange}
In the {\sf min normalized range sum} problem, the range  of each segment is
divided by a monotone increasing function of the number of elements in that segment. The problem is to find a partition, $(S,\bar{S})$ that attains the minimum for the problem:

\begin{eqnarray*}
\min_{\varnothing \subset S \subset V} \frac{range(S)}{f(|S|)}
+\frac{range(\bar{S})}{f(|\bar{S}|)}
 .\end{eqnarray*}

This can be rewritten as:
\begin{eqnarray*}
\min_{\varnothing \subset S \subset V} \frac{max(S)-min(S)}{f(|S|)}
+\frac{max(\bar{S})-min(\bar{S})}{f(|\bar{S}|)}
 .\end{eqnarray*}
We note that since the values of $max(S),min(S),max(\bar{S}),
min(\bar{S})$ are all in the set $ \{\alpha_1,\alpha_2,...,\alpha_n\}$, all possible
combinations may be enumerated in polynomial time. The proof of Lemma \ref{lem:genoverlap}, stated next, is similar to that of Lemma \ref{lem:n_in_S_bar}:
\begin{lem} \label{lem:genoverlap}
There exists an optimal solution to the {\sf min normalized range sum} problem with $v_1\in S$ and $v_n\in\bar{S}$.
\end{lem}
\begin{proof}
If $v_1\in \bar{S}$, then $S$ and $\bar{S}$ may be interchanged with no change
to the value of the objective function.  Due to monotonicity of $f(.)$, and the non-emptiness requirement on $S$, we know $f(n-1) \geq f(|S|)$. Therefore, when $v_1, v_n \in S$, the objective
value can be no lower than $\frac{\alpha_n - \alpha_1}{f(n-1)}$. This is attained
for instance for the solution $S=V\backslash v_i$ and $\bar{S}=v_i$ for
$1<i<n$. However the solution $S'=V\backslash v_n$ and $\bar{S}'=v_n$ has an
objective value that is at most $\frac{\alpha_{n-1} - \alpha_1}{f(n-1)}$ with
$v_n\in\bar{S}'$.
\end{proof}

A generalization of Lemma \ref{lem:maxS} states that in an optimal solution
the two intervals representing $range(S)$ and $range(\bar{S})$ do not overlap:
\begin{lem} \label{lem:genoverlap2}
For any feasible solution $(S,\bar{S})$ to the {\sf min normalized range sum}
problem with $max(S)>min(\bar{S})$, there is a feasible solution
$(S',\bar{S}')$ with $max(S')\leq min(\bar{S'})$ and a lower objective function
value.
\end{lem}
\begin{proof}
For a feasible solution $(S,\bar{S})$ with $max(S)>min(\bar{S})$, define
$S'=\{v_1,...,v_{|S|}\}$ and $\bar{S}'=\{v_{|S|+1},...,v_n\}$.  Note that since
$|S|=|S'|$ and $|\bar{S}|=|\bar{S}'|$, the denominators stay the same. However, $range(S)\geq range(S')$ and
$range(\bar{S})\geq range(\bar{S}')$, so the objective function value is lower
for $(S',\bar{S}')$.
\end{proof}

This implies that for an optimal solution $(S,\bar{S})$ with $max(S)=\alpha_i$,
$min(\bar{S})=\alpha_{i+1}$.  It is therefore sufficient to enumerate the $n-1$ non-overlapping bipartitions in order to solve the {\sf min normalized
range sum} problem.  With the element index $i^*$ so that:
\begin{eqnarray*}
i^*=\arg\min_{i=1..n-1} \frac{\alpha_i - \alpha_1}{f(i)} + \frac{\alpha_n - \alpha_{i+1}}{f(n-i)},
\end{eqnarray*}
the optimal solution is $S=\{v_1,...,v_{i^*}\}$ and $\bar{S}=
\{v_{i^*+1},...,v_n\}$ with an objective function value of $\frac{\alpha_{i^*}
- \alpha_1}{f(i^*)} + \frac{\alpha_n - \alpha_{i^*+1}}{f(n-i^*)}$.
The complexity of this algorithm is $O(n)$.

\section{Min range cut}\label{Prob2}
The {\sf min range cut} problem is to find the partition $(S,\bar{S})$ such
that:
\begin{eqnarray*}
\min_{\varnothing \subset S \subset V} range(S)+range(\bar{S})+C(S,\bar{S}).
 \end{eqnarray*}
As a result of adding the cut to the objective function, some of the properties from the previous sections no longer apply.  For example, we may not
assume that the elements with the largest and smallest values, $v_{1}$ and
$v_{n}$, are in different clusters as the weight on the edge between them,
$w_{1,n}$, could be infinite hence forcing them to be into the same
cluster.

We associate with any partition into $S$ and $\bar{S}$, two respective intervals, $I(S)=I_1=[min (S),max(S)]$ and $I(\bar{S})=I_2=[min (\bar{S}),max(\bar{S})]$.  Since $(S,\bar{S})$ is a partition, the endpoints of these two intervals are four distinct values such that two of the four endpoints must be $\alpha _1$ and $\alpha _n$, and the remaining two endpoints, $\alpha _p$ and $\alpha _q$, are selected from
$\{ \alpha _2,\ldots,\alpha _{n-1} \}$.  We say that the pair of intervals $I_1$ and $I_2$ is feasible if $\{ \alpha _1,\alpha _2,\ldots,\alpha _{n} \} \subseteq I_1\cup I_2 $.
Given $\alpha _p$ and $\alpha _q$ for $2\leq p<q \leq n-1$, then, unless $p=q-1$, there is one feasible choice of the two intervals, as  $[\alpha_1,\alpha _q]$ and $[\alpha_p,\alpha _n]$. If $p=q-1$ then both choices of the two intervals: the pair $[\alpha_1,\alpha _{p+1}]$, $[\alpha_p,\alpha _n]$, and the pair $[\alpha_1,\alpha _{p}]$, $[\alpha_{p+1},\alpha _n]$, are feasible in that their union contains all the values.  For the given selection of $\alpha _p$ and $\alpha _q$  there is a second choice of the interval pair $[\alpha_1,\alpha _n]$ and $[\alpha_p,\alpha _q]$.  It follows that there are $O(n^2)$ feasible choices of $I_1$ and $I_2$ pairs.

Recall that $S$ a subset of interval $I$ is denoted by $S\subseteq I$.
The {\sf min range cut} problem is equivalent to the problem of finding a pair of feasible intervals
minimizing the objective.  We call this problem the {\sf min interval-range cut} problem:

\begin{eqnarray*}
\min _{I_1,I_2 \mbox{ feasible}}\min_{\emptyset \subset S \subseteq I_1, \emptyset \subseteq \bar{S}\subseteq I_2} range(I_1)+range(I_2)+C(S,\bar{S}).
 \end{eqnarray*}
The {\sf min interval-range cut} problem can be equivalently rewritten as:
$$\min _{I_1,I_2 \mbox{ feasible}} range(I_1)+range(I_2)+\min_{\emptyset \subset S \subseteq I_1, \emptyset \subset \bar{S}\subseteq I_2}C(S,\bar{S}).$$

Any feasible solution to {\sf min interval-range cut} is a feasible solution for the {\sf min range cut} problem and vice versa.  It is possible that for a feasible pair of intervals, the objective value of the interval-range cut with an implied pair of sets  $S$ and $\bar{S}$ has the range of $S$ or the range of $\bar{S}$ strictly greater than the range of the respective intervals and therefore larger than the respective objective value for range cut.  This occurs when either one of the sets $S$ and $\bar{S}$ in the optimal solution for $I_1$ and $I_2$, is strictly contained in either $I_1$ or $I_2$ and hence has smaller range.  In that case the selection of $I_1$, $I_2$, while feasible, is not optimal, as there exists another selection pair of feasible intervals $I'_1$, $I'_2$, such that $I'_1 \subseteq I_1$, and $I'_2 \subseteq I_2$, with a strictly better objective value.

\begin{prop} The optimal solution of {\sf min range cut} problem is identical to the optimal solution of {\sf min interval-range cut} problem.
\end{prop} \label{prop:I}
\begin{proof} Let $S^*$ and $\bar{S^*}$ be an optimal solution to {\sf min range cut}.  Then, $I^* _1= [\alpha _1, max (S^*)]$ and $I^* _2= [min (\bar{S^*}, \alpha _n]$ are the arguments of the optimal solution to {\sf min interval-range cut}.
\end{proof}


In order to find an optimal solution for {\sf min range cut}, it is therefore sufficient to enumerate all $O(n^2)$ pairs of feasible intervals, and for each pair of intervals $I_1$ and $I_2$, to solve the respective min-cut problem $\min_{\emptyset \subset S \subseteq I_1, \emptyset \subset \bar{S}\subseteq I_2}C(S,\bar{S})$.  We now elaborate on how to solve the min-cut problem for a given feasible pair of intervals $I_1$ and $I_2$. 

The elements of $V$ are partitioned into $(V_S,V_{\bar{S}},V_F)$, as follows:
\begin{eqnarray*}
V_F& =&\{ i\in V| i\in I_1 \cap I_2\}\\
V_S&=&\{ i\in V| i\in I_1\setminus V_F \}\\
V_{\bar{S}}&=&\{ i\in V| i\in I_2\setminus V_F \}.
\end{eqnarray*}
By construction, any solution to the range cut problem for this given pair of intervals satisfies $V_S\subseteq S$ and $V_{\bar{S}} \subseteq \bar{S}$.  Because the endpoints of the two intervals are distinct, both $V_S$ and $V_{\bar{S}}$ are non-empty.  It remains to partition $V_F$ and allocate its elements to $S$ and $\bar{S}$ so that the cut value is minimized.  Namely,
 $$\min_{\emptyset \subset S \subseteq I_1, \emptyset \subset \bar{S}\subseteq I_2}C(S,\bar{S})= \min _{S \subseteq V_F}C(S,\bar{S}).$$

In cases where $V_F=\emptyset$ the partition into $S$ and $\bar{S}$ is pre-determined and there is no need to solve the min-cut problem.  This happens for the pair of intervals  $[\alpha_1,\alpha _{p}]$, $[\alpha_{p+1},\alpha _n]$.  We next construct a graph $G'$ in which a minimum $s,t$-cut partition provides an optimal solution to the range cut problem for a given feasible pair of intervals.

Consider the graph $G=(V,E)$  with the edge weights $w_{ij}$ associated with each edge $[i,j]\in E$.  Let the graph $G'=(\{ s, t\} \cup V_F, E_{st}) $ be constructed as follows:  The nodes corresponding to $V_S$ are ``shrunk" into a source node $s$, and the nodes corresponding to  $V_{\bar{S}}$ are ``shrunk" into a sink node $t$.
The graph $G'$ is illustrated in Figure \ref{fig:cut_graph}.  It is easy to see that shrinking a node $i$ with the source node $s$ is equivalent, in terms of min-cut in the resulting graph, to adding an arc $(s,i)$ of infinite capacity from $s$ to $i$.  Similarly, shrinking a node $j$ with the sink node $t$ is equivalent to adding an arc $(j,t)$ of infinite capacity from $j$ to $t$.
\begin{figure}[ht]

  \begin{center}
  \scalebox{0.8}{
      \epsfig{figure =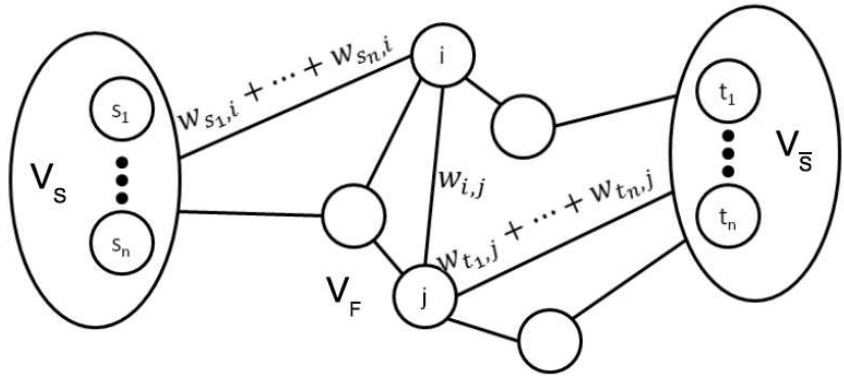}
      }
      \end{center}

  \caption{The $s,t$ graph $G'$ for  $(V_S,V_{\bar{S}},V_F)$ induced by the intervals pair $I_1$, $I_2$. }\label{fig:cut_graph}
\end{figure}

A minimum $s,t$-cut in the graph $G'$ provides an optimal partition of $V_F$ into $S$ and $\bar{S}$.
Therefore, the {\sf min interval-range cut} problem can be solved by enumerating all possible feasible pairs of $I_1$ and $I_2$ and for each finding the min $s,t$-cut partition as described above.  Among all enumerated possibilities, the partition with the lowest objective value is the optimal solution to the {\sf min interval-range cut} problem.

Let ${T(m,n)}$ be the complexity of a minimum $s,t$-cut procedure, on a graph with $m$ arcs and $n$ nodes.  A straight-forward implementation of the algorithm would require $O(n^2)$ calls to the minimum $s,t$-cut procedure, one for each pair of intervals' selection, for a total complexity of $O(n^2\cdot T(m,n))$.

It is shown next that this complexity can be reduced by a factor of $O(n)$ using a {\em parametric cut} procedure:
The parametric flow, or parametric cut, problem is defined on a {\em parametric graph} where the source adjacent capacities and the sink adjacent capacities are functions of a parameter; the source adjacent capacities are monotone non-decreasing in the parameter; and the sink adjacent capacities are monotone non-increasing in the parameter.  All other arcs in the graph have fixed capacities.  The parametric flow problem, and the parametric cut problem which is the focus here, is to solve the maximum flow problem for a list of $q$ parameter values, or for all values of the parameter within an interval of length $U$.  The parametric cut problem for a list of $q$ values was shown to be solved in the complexity of a single cut plus $O(q n)$ for the Push-relabel and Hochbaum's PseudoFlow (HPF) algorithms \cite{gallo, Hoc08}.  Both these max-flow min-cut algorithms have complexity $T(m,n)= O(mn \log {\frac{n^2}{m}})$. (There are other implementations of these algorithms with different complexities as well, e.g.\ $O(n^3)$, \cite{GT,HO13}.)  Therefore the complexity of solving the parametric cut for a sequence of $q$ parameter values is $O(T(m,n)+q n)$, where the term $qn$ accounts for the updates of the source and sink adjacent capacities.

\noindent
{\bf Comment:} The complexity of solving the parametric cut (or flow) problem for all parameter values in an interval of length $U$ was shown in \cite{gallo, Hoc08} to be $O(T(m,n)+ n \log U/\eps)$ for values determined with $\eps$ accuracy (within a $\L _{\infty}$ distance of $\eps$ from the optimal solution).  Although the text in \cite{gallo} claims the complexity of the algorithm to be $O(T(m,n))$, the actual complexity is as stated here{\footnote {The term $n \log U$ cannot be removed from the complexity as proved in \cite{Hoc94,Hoc07,Hoc08}.  Therefore the complexity must depend on the term $\log U$ and hence cannot be strongly polynomial.}}

Consider the parameter list $\{ 1,2,\ldots ,q\}$, and an $s,t$ graph $G$ with arc (or edge) capacities $\{ w^{\ell}_{ij}\}$ that are functions of a parameter $\ell$. An $s,t$ graph $G$ is a parametric graph if the arc capacities are functions of a parameter, and satisfy, for $i,j \neq s,t$:
\begin{eqnarray*}
w^{\ell} _{si} & \leq &w^{\ell +1} _{si} \\
w^{\ell} _{jt} & \geq &w^{\ell +1} _{jt}  \\ 
w^{\ell} _{ij} & = &w^{\ell +1} _{ij}.  
\end{eqnarray*}
Let the parametric graph $G$ with arc capacities $\{ w^{\ell}_{ij}\}$ for a given value of $\ell$ be denoted by $G_{\ell}$.  The parametric graph procedure takes as input, the graph $G$ with arc weights $w_{ij}$ and the parametric sets of source and sink adjacent arcs' capacities:
\\
{\sc parametric cut}$(G, \{ w^{\ell} _{si},w^{\ell} _{jt}\}_{i,j\in V\setminus \{ s,t\} }, \ell=1,\ldots q))$\\
The procedure outputs $S_{\ell}$, for $\ell =1,\ldots ,q$, so that $(S_{\ell}, \bar{S_{\ell}})$ is a minimum cut in $G_{\ell}$.\footnote{The code for HPF {\sc parametric cut} is available at \url{http://riot.ieor.berkeley.edu/Applications/Pseudoflow/parametric.html}.}


Let the graph $G_j$ be generated from graph $G_{j-1}$ by shrinking one node $v$ in $G_{j-1}$ with the source. Since the shrinking of node $v$ with the source is equivalent to adding an arc $(s,v)$ of infinite capacity, then this process increases the capacity of the arc adjacent to source $(s,v)$ from a finite value to infinity, while all other capacities remain constant.  Hence the sequence of graph generated by shrinking with the source, one node at a time, form a parametric graph.  Also, the sequence of graphs generated by shrinking one node at a time with the sink node $t$, done by adding an infinite capacity arc between the node and $t$, form a parametric graph with sink adjacent capacity monotone non-decreasing and source adjacent that are constant, and thus non-increasing.  The complexity of solving the min-cut for such sequence of $q$ graphs, $G_{\ell},\ldots , G_{\ell +q}$, with parametric cut is therefore $O(T(m,n))$ since $q<n$, and $O(n^2)$ complexity is dominated by $T(m,n)$.


In the following theorem we show that  Algorithm \ref{alg:mrc} solves the min range cut problem in $O(nT(m,n))$ steps.

\begin{algorithm}[hbt!]

  \caption{Min Range Cut Algorithm}
  \label{alg:mrc}
  \begin{algorithmic}[1]
   \STATE \textbf{INPUT:} An ordered array $\{\alpha_1,...,\alpha_n\}$
       \STATE \textbf{OUTPUT:} $(S^*,\bar{S}^*)$, and the corresponding objective value  $z^*$
    \STATE \textbf{begin:}
    \STATE $\min(I_1):=\alpha_1$; $z^*=\infty$

\STATE $\max(I_2):=\alpha_n$   \  \texttt{ \small{\%}checking all possible $\max(I_1)$ and $\min(I_2)$)}

\FOR {$i= 1,...,n-1$}
\STATE \hspace{0.5cm} $\max(I_1) := \alpha_i $.
\STATE \hspace{0.5cm} Create graph $G(i)=(V(i),E(i))$: set $v_1$ as source node $s$, shrink  $v_{i+1},...,v_n$ into \\
\hspace{0.5cm} sink $t$; arc capacities for $(p,q)\in E(i)$ are $w_{pq}^{(1)}$.
\STATE \hspace{0.5cm} \textbf {for}~ {$j= 2,...,i-1$}~\textbf{do}

\STATE \hspace{1cm} $\min(I_2) :=\alpha_j $;
\STATE \hspace{1cm}\texttt{ \small{\%}Shrink node $v_{j}$ into $s$}
\[ w_{pq}^{(j)}= \left\{ \begin{array}{ll}
      w_{pq}^{(j-1)} & \mbox{ if     }\  (p,q)\neq (s,v_j) \\
      \infty & \mbox{ if     }\  (p,q)\neq (s,v_j).
                \end{array}
\right. \]
\STATE \hspace{0.5cm} \textbf{end for}
\STATE \hspace{0.5cm} Call {\sc parametric cut}$(G(i), \{ w^{(j)} _{sp},w^{(j)} _{qt}\}_{p,q\in V(i)\setminus \{ s,t\} },j=2,\ldots i-1))$; return $S_2,\ldots ,S_{i-1}$
\STATE \hspace{0.5cm} set $j^*= arg \min _j z(j) =(\alpha _i -\alpha _1)+(\alpha _n -\alpha _j) +C(S_j,\bar{S_j})$. 
\STATE \hspace{0.5cm}  \textbf{if}~{$z(j) < z^*$}~\textbf{then}
\STATE\hspace{1cm} $(S^*,\bar{S}^*):= (S_{j^*},\bar{S}{j^*})$
\STATE \hspace{1cm} $z^*:= z(j)$
\STATE \hspace{0.5cm} \textbf{end if}
\ENDFOR

\STATE  $\max(I_1):=\alpha_n$   \  \texttt{ \small{\%} for $I_1=[\alpha _1,\alpha _n]$ checking all possible $\max(I_2)$ and $\min(I_2)$}
\FOR{$i = 2,...,n-1$}
\STATE \hspace{0.5cm} $\max(I_2) := \alpha_i $.
\STATE \hspace{0.5cm} Create graph $G(i)=(V(i),E(i))$: set  $v_i$ as source node $s$, shrink  $v_1,v_{i+1},...,v_n$ into \\
\hspace{0.5cm} sink $t$; arc capacities for $(p,q)\in E(i)$ are $w_{pq}^{(1)}$.

\STATE \hspace{0.5cm} \textbf{for}~ {$j=3,...,i$}~\textbf{do}
\STATE \hspace{1cm} $\min(I_2) := \alpha_j $.
\STATE \hspace{1cm} \texttt{ \small{\%}Shrink node $v_{j-1}$ into $s$}
\[ w_{pq}^{(j)}= \left\{ \begin{array}{ll}
      w_{pq}^{(j-1)} & \mbox{ if     }\  (p,q)\neq (s,v_{j-1}) \\
      \infty & \mbox{ if     }\  (p,q)\neq (s,v_{j-1}).
                \end{array}
\right. \]
\STATE \hspace{0.5cm} \textbf{end for}
\STATE \hspace{0.5cm} Call {\sc parametric cut}$(G(i), \{ w^{(j)} _{sp},w^{(j)} _{qt}\}_{p,q\in V(i)\setminus \{ s,t\} },j=3,\ldots i))$; return $S_2,\ldots ,S_{i-1}$
\STATE \hspace{0.5cm} set $j^*= arg \min _j z(j) =(\alpha _n -\alpha _1)+(\alpha _i -\alpha _j) +C(S_j,\bar{S}_j)$. 
\STATE \hspace{0.5cm}  \textbf{if}~{$z(j)< z^*$}~\textbf{then}
\STATE \hspace{1cm} $(S^*,\bar{S}^*):=(S,\bar{S})$
\STATE\hspace{1cm}  $z^*:= z(j)$
\STATE \hspace{0.5cm} \textbf{end if}
\ENDFOR
    \STATE \textbf{end}
\end{algorithmic}
\end{algorithm}

\begin{thm} Algorithm \ref{alg:mrc} solves the {\sf min range cut problem} in $O(n{T(m,n)})$ steps.
\end{thm}

\begin{proof}
The correctness of the algorithm follows from its enumeration of all possible endpoints of intervals $I_1$ and $I_2$ and for each find the min-cut partition of the elements in the overlap of the two intervals.

The enumeration of all endpoints is done in two parts:  Lines 4-18 deal with pairs of feasible intervals $I_1=[\alpha _p ,\alpha _n]$ and $I_2=[\alpha _1 ,\alpha _q]$, and lines 20-34 deal with pairs of feasible intervals $I_1=[\alpha _1 ,\alpha _n]$ and $I_2=[\alpha _p ,\alpha _q]$ where $I_1\cap I_2 = I_2$.

The complexity of the Min Range Cut Algorithm is dominated by the two {\bf for} loops, in steps $6$ and $21$, each of which
consists of calling at most $n-2$ times for the parametric min-cut procedure. Each call for parametric cut has complexity of $O(T(m,n))$.  Therefore the complexity of the entire algorithm is $O(nT(m,n))$.
\end{proof}

\section{Min normalized range cut problem}\label{Prob3}

Unlike min range sum, min normalized range sum and min range cut, the {\sf min
normalized range cut} problem is not polynomial time solvable.  We demonstrate here the NP-hardness of the problem. In this section we
slightly abuse terminology by referring to optimization problems as NP-complete meaning that their decision version is NP-complete (the correct term for optimization problems is NP-hard).  All problems we address are clearly in NP and we will omit explicitly showing so.
\begin{thm}
The {\sf min normalized range cut} problem,  $\min\limits_{\varnothing \subset S \subset V} \frac{range(S)}{|S|}+
\frac{range(\bar{S})}{|\bar{S}|}+C(S,\bar{S})$, is NP-complete.
\end{thm}

\begin{proof}
To prove NP-completeness, we use Karp reductions in two steps to show that {\sf min normalized
range cut} $\propto_P$ {\sf min inverse set size cut} $\propto_P$ {\sf balanced
cut}, where $Q_1 \propto_P Q_2$ means that problem $Q_1$ is at least as hard as problem $Q_2$, and equivalently, that problem $Q_2$ is polynomial time reducible to problem $Q_1$.   We first introduce these problems and then provide the NP-completeness proof.
The {\sf balanced cut} problem is a known NP-complete problem of finding the minimum 2-cut where the two
sets in the bipartition are of equal size, \cite{GJS}. For a graph $G=(V,E)$ on even number of nodes $n=|V|$, the balanced cut problem is formulated as:
\begin{eqnarray*}
\min_{S\subset V, |S| = \frac n2} C(S,\bar{S}).
\end{eqnarray*}
The {\sf min inverse set size cut} problem is used here as an
intermediary problem.  This problem is defined as:
\begin{eqnarray*}
\min\limits_{\varnothing \subset S \subset V}  \frac{1}{|S|}+\frac{1}{|\bar{S}|}+C(S,\bar{S}).
\end{eqnarray*}
Part 1)  We first show that {\sf balanced cut} is reducible to {\sf min inverse set size cut}: Given an instance of balanced cut defined on $G=(V,E)$  with edge weights $w_{ij}$.  Define a new, scaled graph, in which the edge weights are $w'_{ij}=\frac{w_{ij}}{M}$, for some large number $M$.  A suitable choice of $M$ is $M=w_{\max} n^4$, where $w_{\max}=\max _{[i,j]\in E} w_{ij}$.   We note that the minimum cut partition in a scaled graph is the same as the minimum cut partition as in the original graph, and the capacity of the scaled cut is $\frac{1}{M}$ times the capacity of the cut in the original graph.
The {\sf min inverse set size cut} on this scaled graph is:
\begin{eqnarray*}
\min\limits_{\varnothing \subset S \subset V}  \frac{1}{|S|}+\frac{1}{|\bar{S}|}+\frac{C(S,\bar{S})}M.
\end{eqnarray*}
Since there are at most $O(n^2)$ arcs in a cut, then for our choice of the value of $M$, $\frac{C(S,\bar{S})}M$ is at most $O(\frac{1}{n^2})$. Consequently the first two terms dominate the cut value term,
$\frac{1}{|S|}+\frac{1}{|\bar{S}|} \gg \frac{C(S,\bar{S})}M$ by a factor of $O(n)$.  Thus the
optimal solution will necessarily minimize:
\begin{eqnarray*}
\min\limits_{\varnothing \subset S \subset V}  \frac{1}{|S|}+\frac{1}{|\bar{S}|},
\end{eqnarray*}
for which the minimum is attained for $|S|=|\bar{S}|=\frac n2$.  Among the solutions that minimize the first two terms, the {\sf min inverse set
size cut} problem minimizes the term $\frac{C(S,\bar{S})}M$, resulting in a minimum balanced cut.   Thus, {\sf min inverse set size cut} problem is NP-complete.
\begin{figure}[ht]
  \begin{center}
\subfigure[Original graph]{\epsfig{figure=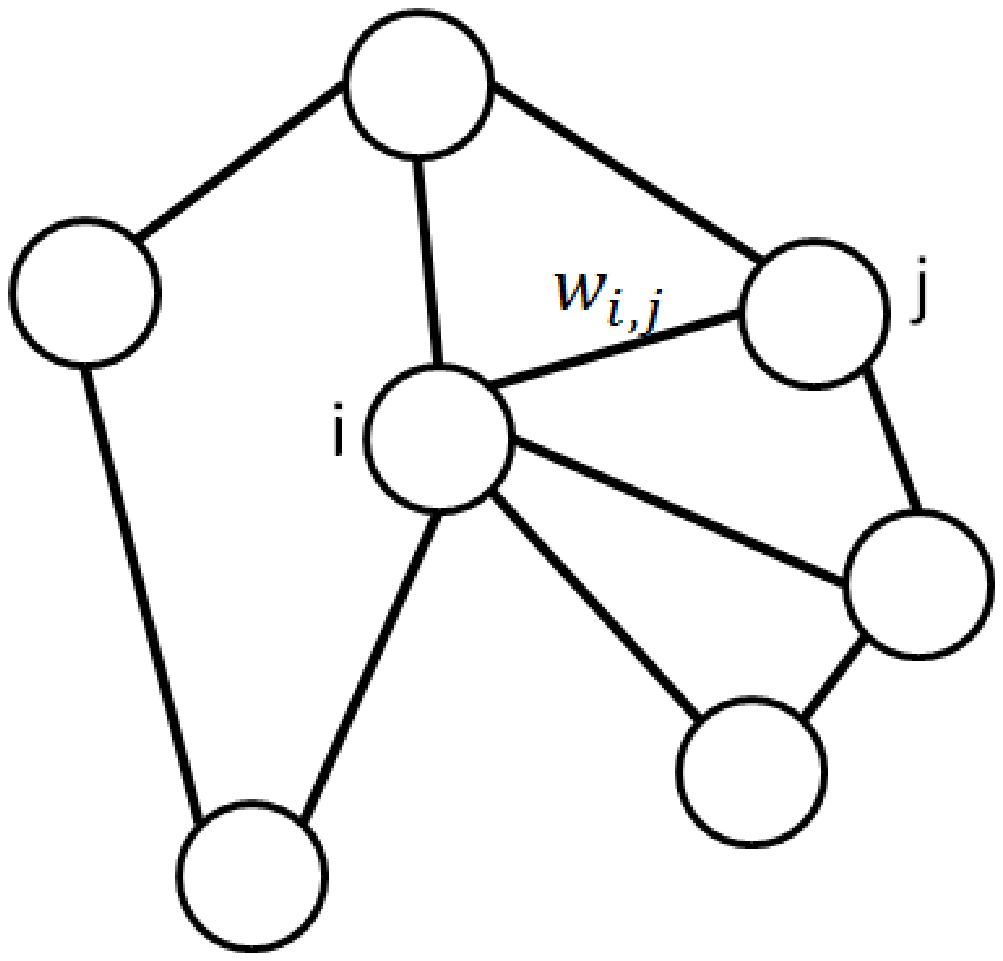, width =1.8in}}
\quad \quad \subfigure[New graph]{\epsfig{figure=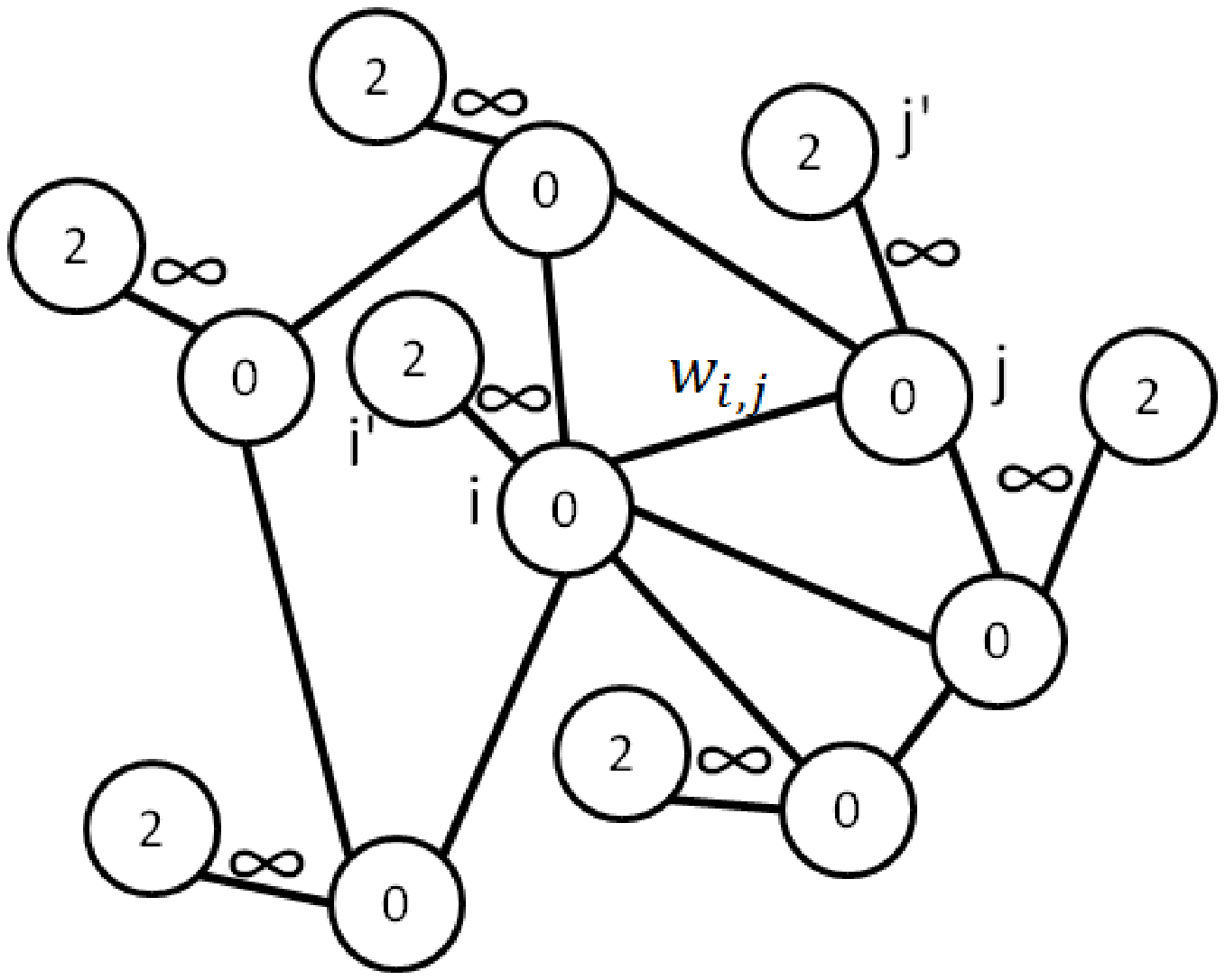, width = 2.2in}}
      \end{center}
  \caption{The original graph and the new graph in part 2's reduction.  Numbers inside the nodes of
  the new graph are the values of the respective elements.}\label{nppic}
\end{figure}

Part 2)  We now demonstrate that {\sf min inverse set size cut} is reducible to {\sf min normalized range cut}:
For a given problem instance of the {\sf min inverse set size cut} problem with $n$
nodes and $m$ edges, we construct a new graph with $2n$ nodes and $m+n$ edges:
Each original node $v$ of $V$ is connected to
a new node $v'$ with an edge of capacity $\infty$.  All original nodes are assigned
a value of $0$ and all new nodes are assigned a value of $2$.  See Figure
\ref{nppic}.
The presence of edges of infinite capacity guarantees that
the range of both $S$ and $\bar{S}$ for any finite cut partition is exactly two, as otherwise the cut would have to
have an infinite value.  Also, for any finite cut partition in the original graph, the corresponding partition in the new graph has double the cardinality of $S$ and $\bar{S}$ and the same cut capacity. Thus, the solution to this new problem using {\sf min normalized range cut} is
exactly the solution to the original instance of the {\sf min inverse set size
cut} problem.
\end{proof}

\section{Range segmentation in $k$-partitions}\label{k-seg}

In this section the bipartition results are extended to $k$-partitions. Specifically, we devise polynomial time algorithms for {\sf min k range sum}, {\sf min max k range}, {\sf min k-normalized range sum} and show that  {\sf min k-normalized range cut} is NP-complete.  For any fixed value of $k$, the {\sf min k range cut} is shown here to be polynomial time solvable, as is the case for $k=2$.  But for arbitrary value of $k$ we prove here that the {\sf min k range cut} is NP-hard.  This follows since the problem generalizes the  min $k$-cut which is NP-hard for arbitrary $k$, \cite{GH}.
\subsection{Min $k$ range sum problem}\label{prob1k}

The {\sf min $k$ range sum problem} is to find a partition $(S_1,...,S_k)$, of $V$, that minimizes $
\min \ \sum_{i=1}^k range(S_i)$.
This problem can be solved in linear time with Algorithm \ref{alg:mkrs} as shown next.

\begin{algorithm}[h!]

  \caption{Min $k$ Range Sum Algorithm}
  \label{alg:mkrs}
  \begin{algorithmic}[1]
   \STATE \textbf{INPUT:} An ordered array $\{\alpha_1,...,\alpha_n\}$
          \STATE \textbf{OUTPUT:}$(S^*_1,S^*_2,\ldots,S^*_k)$, such that $\cup _{i=1}^k \{ S^*_i\}= V$ and the objective value  $z^*$.
    \STATE \textbf{begin:}



\FOR{$i= 1,\ldots,n-1$}

\STATE \hspace{0.5cm}$g_i := \alpha_{i+1}-\alpha_i$

\ENDFOR

\STATE Let {$i_1,\ldots,i_{k-1}$} be the indices of the largest $k-1$ elements of $\{g_i\}_{i=1}^n$.

\STATE  $(S^*_1,S^*_2,\ldots,S^*_k):= (\{\alpha_1, \ldots,\alpha_{i_1}\},\{\alpha_{i_1+1},\ldots,\alpha_{i_2}\},\ldots,\{\alpha_{i_{k-1}+1},\ldots,\alpha_{n}\} )$.

\STATE $z^* = \sum_{j=1}^k {range(S^*_j)}=\sum_{j=1}^k (\alpha_{i_j}-\alpha_{i_{j-1}+1})$, for $\alpha_{{i_0}+1}=\alpha_1$.

 \STATE \textbf{end} 

\end{algorithmic}
\end{algorithm}

\begin{prop} 
Algorithm \ref{alg:mkrs} solves the {\sf min k range sum} problem in $O(n)$ steps. \end{prop}
\begin{proof}
In an optimal solution $v_1 \in S_1$ and $v_n \in S_k$, and the sets $(S_1,\ldots,S_k)$ are non-overlapping. These facts follow from the same arguments used in Lemma \ref{lem:n_in_S_bar} and Lemma \ref{lem:maxS} and is omitted for brevity. It remains to determine the boundaries between the non-overlapping segments on the real line. This is equivalent to identifying the $k-1$ largest gaps between consecutive values of $\alpha$.  This can be done by finding first the $(k-1)^{st}$ median in the set of $n-1$ gaps, in linear time $O(n)$, using the algorithm of Blum et al.\ \cite{blum-median}. Once this median is found, the set of gaps is scanned once to mark all the gaps that have value greater or equal to that of the $(k-1)^{st}$ median.  This produces the $k-1$ largest gaps in linear time as shown in Algorithm \ref{alg:mkrs}.  These gaps separate the sets in the partition with the smallest sum of ranges.
\end{proof}

\subsection{Min max k range}
The {\sf min max k range} problem for $k>2$ is solved differently and with significantly higher complexity than the case of $k=2$ (in Section \ref{minmax}).  Recall that the {\sf min max k range} is to find a partition $(S_1,...,S_k)$ of $V$, that minimizes $\min\limits_{(S_1,\ldots,S_k)}\left( \max_{i\in \{1..k\}} range(S_i)\right)$.

As before, it is established, with the same arguments as in Lemma \ref{lem:n_in_S_bar} and Lemma \ref{lem:maxS}, that there is an optimal solution in which $v_1 \in S_1$ and $v_n \in S_k$, and that the intervals containing the sets $(S_1,\ldots,S_k)$ are non-overlapping. Specifically, each set $S_j$ in an optimal partition contains consecutive elements and is of the form $\{ i_{j-1}+1,\ldots , i_j\}$ with range $\alpha _{i_j}-\alpha _{i_{j-1}+1}$.

%
%

\begin{prop} Algorithm  \ref{alg:mmkg} solves the {\sf min max k range} problem in $O(n\log^3 n)$ steps. \end{prop}
\begin{proof}
The algorithm works by guessing one value for the max range at a time, and then conducting a {\em feasibility check} to verify whether there is a feasible solution for that value.  The smallest value of a guessed range for which there is a feasible solution is the optimal range value.  A natural way of implementing such an algorithm is by using binary search on the $n\choose2$ possible range values.  Each possible range value is of the form $\alpha _q -\alpha _p$ corresponding to a pair $p,q$ such that $1\leq p <q\leq n$.  Let the set of all possible range values be $\mathcal{D}=\{\alpha_q-\alpha_p| q>p, p=1,\ldots,n-1, q=p+1, \ldots, n\}$.  Note that we disregard the trivial case where the optimal solution is $0$.  The trivial case happens when the number of distinct scalar values is at most $k$.
Since there are up to $n\choose2$ possible range values, these can be sorted, in $O(n^2 \log n)$ time (note that $\log {n\choose2}$ is $O(\log n)$). Let the sorted values in $\mathcal{D}$ be $d_1 \geq d_2 \ldots \geq d_{n\choose2}$.

Consider the feasibility check for a given guessed value for the max range, $z$.  To verify feasibility we first scan the values of $\alpha _1, \alpha _2 \ldots \alpha _n$, for the largest index $j_1$ so that  $\alpha _{j_1} - \alpha _1 \leq z$. If $\alpha _{2} - \alpha _1 > z$ then $j_1=1$.  The interval $I_1=[1,\alpha _{j_1}]$ is then the first of up to $k$ intervals representing the $k$-partition of the set of values $\alpha _1, \alpha _2 \ldots \alpha _n$.  Next we scan the values of $\alpha _{j_1+1}, \alpha _{j_1+2} \ldots $, for the largest index $j_2$ such that  $\alpha _{j_2} - \alpha _{j_1 +1} \leq z$.   Again, if $\alpha _{j_1+1}-\alpha _{j_1}>z$, then $j_2=j_1$, corresponding to an interval that contains only one element, which is of range $0$.  The interval $I_2$ is then equal to $[\alpha _{j_1 +1},\alpha _{j_2}]$.  This is repeated up to $k$ times or until the last value $\alpha _n$ is reached. If after $k$ repetitions $j_k <n$ then the guessed value $z$ is not feasible and therefore the optimal range value must be greater.  Otherwise the guessed value is feasible and the optimal range value can only be smaller than $z$.   This feasibility check runs in $O(n)$ steps as it scans the values of $\alpha$ at most once each.

We comment that there is an alternative feasibility check on the guessed value $z$ that runs in $O(k\log n)$ steps.  For $\alpha _{j_i}$, the start of the $i$th interval,  we search, using binary search on the set of $\alpha$ values, for the next value of $\alpha$, $\alpha _{j_{i +1}}$, which is the largest while the difference from the current value of $\alpha$ is still less than $z$.  Each such search requires $O(\log n)$ steps, and since there are up to $k$ such intervals the total complexity is $O(k\log n)$. This complexity is faster than $O(n)$ if $k=o ({\frac{n}{\log n}})$, but even then it does not improve the overall complexity since the other steps dominate it, as discussed next.

The optimal  solution to the min max k range problem can then be found using binary search on the sorted sequence  $\mathcal{D}$.  This requires $O(\log n)$ calls for feasibility check for a total complexity of $O(n \log n)$ which is dominated by the complexity of the sorting of $\mathcal{D}$, $O(n^2 \log n)$.   Next we show that the need to sort the distances in $\mathcal{D}$ can be avoided, resulting in substantial speed-up from $O(n^2 \log n)$ down to $O(n\log^3 n)$.

\begin{algorithm}[t!]

  \caption{Min max $k$ range}
  \label{alg:mmkg}
  \begin{algorithmic}[1]
   \STATE \textbf{INPUT:} An ordered array $\boldalpha =\{\alpha_1,...,\alpha_n\}$,  in increasing order. 
       \STATE \textbf{OUTPUT:} $(S^*_1,S^*_2,\ldots,S^*_k)$, and the corresponding objective value $z^*$.

       \STATE \textbf{begin:}

  \STATE $a_{\min} := 1; a_{\max}:= {n\choose2}$;
      \WHILE {$a_{\min} < a_{\max}$}
       \STATE\hspace{0.5cm} $I_i := \emptyset, \forall i=1,\ldots,k$;  
        \STATE \hspace{0.5cm} $m:=  \lfloor {{1\over 2}(a_{\min}+a_{\max})}\rfloor $;
        \STATE  \hspace{0.5cm} $z:=$ $m$th largest value in the set $\mathcal{D}=\{\alpha_q-\alpha_p| q>p, p=1,\ldots,n-1, q=p, \ldots, n\}$; \\
 \hspace{0.5cm}  (determined by the M-algorithm);


     \hspace{0.5cm}   \texttt{\small{\%}Checking whether a feasible solution exists with value $\leq z$:}
        \STATE\hspace{0.5cm} $\min(I_1):=\alpha_1$;
       \STATE \hspace{0.5cm} $\max(I_1) :=$ the largest $\alpha_j$, $j\geq 1$, such that $\alpha_j - \alpha_1 \leq z$;

\STATE\hspace{0.5cm} \textbf{for~}{$i= 2,\ldots,k$}~\textbf{until~}{$\max I_{i-1}=\alpha _n$}~\textbf{do}

 \STATE  \hspace{1cm} $\min(I_i):=$ the smallest $\alpha_j$, such that $\alpha_j > \max(I_{i-1})$;
 \STATE \hspace{1cm} $\max(I_i):=$ the largest $\alpha_j$, such that $\alpha _j \leq \min(I_{i})+z$;

\STATE\hspace{0.5cm} \textbf{end for}
\STATE\hspace{0.5cm} \textbf{if~} {$\alpha_n \in \cup_{i=1}^k\{I_i\}$}~\textbf{then}
\STATE \hspace{1cm} $a_{\max}:=m$;
\STATE\hspace{0.5cm} \textbf{else}
\STATE\hspace{1cm} $a_{\min}:=m+1$;
\STATE\hspace{0.5cm} \textbf{end if}
\ENDWHILE
\\
\texttt{ \small{\%}Computing the output:}
\STATE  $z^*:=z$;
\FOR {$i = 1 \ldots k$}
\STATE\hspace{0.5cm} $S^*_i := \cup_{v_j \in I_i} \{v_j\}$;
\ENDFOR
 \STATE \textbf{end}
 \end{algorithmic}
\end{algorithm}

Megiddo et al.\ \cite{MTZH} devised an efficient algorithm, called here the M-algorithm, for finding the $i$th longest path among the set of all simple paths in a tree with edge weights.  For a tree on $n$ nodes the complexity of the M-algorithm is $O(n \log^2 n)$.  Note that the number of different simple paths in a tree is $O(n^2)$, since each simple path is uniquely characterized by its pair of endpoints.
Consider a path graph on nodes $\{ 1,2,\ldots ,n\}$, where all edges are of the form $[i,i+1]$ with weights $\alpha _{i+1}-\alpha _i$ for $i=1,\ldots ,n-1$.  A path graph is obviously a tree and thus the M-algorithm  is applicable to this path graph.  The distance between node $p$ and node $q$, for $p<q$ is then $\alpha_q-\alpha_p$.   Now, instead of sorting the distances in $\mathcal{D}$, we use the M-algorithm to identify the $i$th longest of the potentially feasible ranges which is the value $z$ to be checked for feasibility.  Since each call and feasibility check reduces the number of potentially feasible ranges by a factor of $2$, the total number of calls is $\log {n\choose2}$, which is $O(\log n)$.

Initially the interval of integer indices $[1,{n\choose2}]$ contains the list of the index positions of all the potentially feasible ranges.  At each iteration we find the median value in this interval of indices, without having the sorting available, by calling the M-algorithm. If this median range value is feasible, then we conclude that the min max feasible range can be only smaller and thus resides in the list of indices smaller or equal to the median.  Otherwise it resides in the list of larger indices.  Initially the endpoints of the interval of indices are  $a_{\min}=1$ and $a_{\max}={n\choose2}$.  At each iteration the M-algorithm finds the median range value $z$ in the interval, which is the $\lfloor {{1\over 2}(a_{\min}+a_{\max})}\rfloor $th longest in the original list.   If $z$ is feasible then $a_{\max}$ is updated to be equal to this median index, otherwise $a_{\min}$ is updated to be equal to this median index plus $1$.  The length of the interval is hence reduced by a factor of $2$ at each iteration, thus requiring at most $\log {n\choose2}$ iterations of the procedure. 

At each of the $O(\log n)$ calls for the guessed value of the range there is one call for finding the $\lfloor {{1\over 2}(a_{\min}+a_{\max})}\rfloor $th longest value $z$ and one call for feasibility check of $z$.  The first requires $O(n \log^2 n)$ steps and the second requires $O(n)$ steps.  The total complexity is then $O(\log n \cdot (n \log^2 n  +n))$. Thus, the optimal solution to min max $k$ range problem is computed in time $O(n\log^3 n)$. \end{proof}

The pseudocode of the algorithm solving min max $k$ range is given as Algorithm \ref{alg:mmkg}.

\subsection{Min k-range cut}
Recall that the {\sf min k-range cut} problem is,
$\min_{(S_1,\ldots,S_k)}  \ \sum_{i=1}^k range(S_i) + \sum_{i=1}^{k-1} \sum_{j=i+1}^k C(S_i,S_j).$
For the range-cut (2-range cut) problem we reduced the problem to calls, for each configuration of interval partitioning, to a min cut procedure (Section \ref{Prob2}).   The idea here is analogous, exploiting the use of the polynomial time minimum $k$-cut algorithm for $k$ fixed of \cite{GH}.  Firstly we note that for arbitrary $k$ the NP-hardness of the minimum $k$-cut problem, \cite{GH}, implies that {\sf min k-range cut} is NP-hard.  This is easy to show by selecting the range of the values to be very small, by a factor of $n$ at least, than the smallest weights of the graph.  And then the value of any solution is dominated by the value of the $k$-cut partition.

The $k$-cut problem was shown to be solved in polynomial time for fixed $k$ with the algorithm of Goldschmidt and Hochbaum \cite{GH} (GH-algorithm).  The GH-algorithm involves $O(n^ {O(k^2)})$ calls to a min $s,t$-cut procedure.  For $k$ not fixed the problem was shown to be NP-hard. The way the algorithm works is by guessing a set of ``seeds" that must belong to one set in the partition, and a set of ``seeds" that belong to the other $k-1$ sets.  The algorithm calls for a respective min $s,t$-cut for the seeds in the set shrunk into the source set $s$, and the seeds for the other sets shrunk into the sink node $t$.  The resulting source set is then considered to be one set in the $k$-partition and the process then continues, recursively, on the subgraph induced by the sink set, $k-1$ additional times, resulting in a $k$-partition.  It was shown in \cite{GH} that it is sufficient to select seed sets that contain at most $O(k)$ seeds, and thus enumerating all of them takes time that is polynomial for fixed $k$.

Our algorithm for  {\sf min k range cut} is a generalization of the min range cut algorithm, for $k=2$, and similarly works in two steps.  In the first step 
select all possible feasible collections of $k$ intervals, each corresponding to one set $S_i$, determined by the endpoints $[min(S_i),max(S_i)]$ (possibly $min(S_i)=max(S_i)$).  Since the sets $S_i$ form a partition, the endpoints of the respective intervals are distinct.  Hence there are up to $2k-2$ distinct endpoints, and respectively up to $O(k n^ {2k-2})$ interval configurations.  A $k$-intervals configuration is feasible, if all values of $\alpha$ are contained in the union of the intervals. For each interval configuration that is feasible we let the seeds for the i$^{th}$ set in the partition $S_i$ include $min(S_i)$ and $max(S_i)$ and all the nodes that correspond to values in the interval $(min(S_i),max(S_i))$ that are not in any of the other $k-1$ intervals.  The resulting set of seeds is then augmented, if necessary, by the GH-algorithm, which is otherwise employed without change.

The total complexity involves then  $O(k n^ {2k-2})$ calls to the $k$-cut algorithm, each of complexity $O(n^ {O(k^2)}T(m,n))$.  The total complexity of solving {\sf min k-range cut} problem is hence $O(n^ {O(k^2)}T(m,n))$.
%

\subsection{Min k-normalized range sum} \label{sec:prob2k}

The {\sf min k-normalized range sum} problem  is to find a $k$-partition, $(S_1,\ldots,S_k)$, to achieve the following objective:
\begin{eqnarray*}
Q(n,k)=\min\limits_{(S_1,\ldots,S_k)}\ \ \sum_{i=1}^k \frac{range(S_i)}{f(|S_i|)}
\end{eqnarray*}
We next present a polynomial time dynamic programming algorithm for the problem:
 \begin{prop} The {\sf min k-normalized range sum} problem is solvable in polynomial time $O(n^2k)$. \end{prop}
\begin{proof}
  As before, it can be proved that there exists an optimal solution in which  $v_1 \in S_1$ and $v_n \in S_k$, and that the the sets $(S_1,\ldots,S_k)$ are non-overlapping. The proof follows the same argument as in Lemma \ref{lem:genoverlap} and Lemma \ref{lem:genoverlap2} and is omitted.

The value of the objective function, $Q(n,k)$ is the minimum \emph{cost} required to partition $n$ elements into $k$ sets.
For the ordered input elements according to $\{\alpha_1<...<\alpha_n\}$, let $Q(p,j)$, be the minimum cost for a partition of the first $p$ elements of the input array into $j$ sets, for $p\in\{1,\ldots,n\}$ and $j \in \{1,\ldots, k\}$ .  We construct a dynamic programming recursion with the boundary conditions:  $Q(j,j) = 0~ \forall j$,  $Q(p,1)= \frac{\alpha_p - \alpha_1}{f(p)}~ \forall p$, and $Q(i,j) = \infty ~ \forall i<j$, with the latter being infeasible and therefore set to $\infty$. The following recursion is used to calculate $Q(p,j)$ for $p,j > 1$, once the values of $Q(p',j')$ have been evaluated for all $p'<p$ and $j'<j$:
\begin{eqnarray*}
Q(p,j) = \min_{\ell \in \{j-1,..,p-1\}} \left(Q(\ell , j-1) + \frac{\alpha_p - \alpha_{\ell+1}}{f(p -\ell )}\right)
\end{eqnarray*}
The rationale for the recursion is that optimal partitioning of $p$ elements into $j$ sets consists, for some value $\ell$,  of an optimal partitioning of elements $1, \ldots, \ell$ into $j-1$ sets and allocating elements $\ell +1, \ldots, p$ into the $j$th set. Since each recursion evaluation is accomplished in at most $O(p)$ steps, and there are $O(kn)$ function evaluations it follows that the optimal solution $Q(n,k)$ is determined with this dynamic programming procedure in complexity of $O(n^2k)$.
\end{proof}

\subsection{Min k-normalized range cut}

In section \ref{Prob3}, we proved this problem to be NP-complete in a bipartition setting. The following theorem follows by simply noticing that the bipartition problem is a special case of {\sf min k-normalized range cut} problem, with $k=2$.
 \begin{thm} The {\sf min k-normalized range cut} problem is NP-complete. \end{thm}

\section{Conclusions} \label{conc}
We introduce here a novel criterion in clustering that seeks clusters with limited \emph{range} of values that characterize each cluster's elements.  We present a family of {\em range-based} clustering objective functions based on commonly
considered goals in clustering and demonstrate that, in general, the range-based optimization problems are easier to solve (complexity-wise) than the corresponding total similarity problems.  The proposed objectives could therefore be a viable alternative to existing
clustering criteria, that are NP-hard, offering the advantage of efficient algorithms. Moreover, the range-based problems are meaningful in clustering applications, such as image segmentation, where the diameter, or range of values associated with objects in each cluster,  should be small.

\section*{Acknowledgement}
The author is grateful to John E.\ Baumler for inspiring discussions that initiated the idea of range-clustering, and to Marjan Baghaie for her contributions to the writing of this paper and pointing out reference \cite{MTZH}. This is to thank the referee for detailed and perceptive comments as well as providing reference \cite{phan15}.  The referee's numerous suggestions contributed significantly to improving the presentation and are much appreciated.


\end{document}